\documentclass[12pt]{amsart}

\usepackage{textcomp}
\usepackage{amsthm}
\usepackage{amssymb}
\usepackage{amsmath}
\usepackage{graphics,graphicx,colortbl}
\numberwithin{equation}{section}
\newcommand{\bea}{\begin{eqnarray}}
\newcommand{\eea}{\end{eqnarray}}
\newcommand{\be}{\begin{eqnarray*}}
\newcommand{\ee}{\end{eqnarray*}}
\newtheorem{theorem}{Theorem}[section]

\newtheorem{corollary}{Corollary}[section]
\newtheorem{definition}{Definition}[section]
\newtheorem{proposition}{Proposition}[section]
\newtheorem{example}{Example}[section]

\begin{document}
\title[Boolean Networks with Multi-Expressions and Parameters]{Boolean Networks with Multi-Expressions and Parameters}
\author[Yi Ming Zou]{Yi Ming Zou}
\address{Department of Mathematical Sciences, University of Wisconsin, Milwaukee, WI 53201, USA} \email{ymzou@uwm.edu}
\subjclass{Primary: 94C10; Secondary: 05C38.} 
\keywords{Boolean networks, gene regulatory networks, multistate, asynchronous, bioinformatics, computational biology}
\maketitle
\date{}

\begin{abstract}
To model biological systems using networks, it is desirable to allow more than two levels of expression for the nodes and to allow the introduction of parameters. Various modeling and simulation methods addressing these needs using Boolean models, both synchronous and asynchronous, have been proposed in the literature.  However, analytical study of these more general Boolean networks models is lagging. This paper aims to develop a concise theory for these different Boolean logic based modeling methods. Boolean models for networks where each node can have more than two levels of expression and Boolean models with parameters are defined algebraically with examples provided. Certain classes of random asynchronous Boolean networks and deterministic moduli asynchronous Boolean networks are investigated in detail using the setting introduced in this paper. The derived theorems provide a clear picture for the attractor structures of these asynchronous Boolean networks.
\end{abstract}
\section{Introduction}
\par
\medskip
A Boolean network of $n$ nodes can be represented by a function from  $F_2^n$ to  $F_2^n$, where $F_2=\{0,1\}$ is the field of two elements, and  $F_2^n$ is the $n$-dimensional space over $F_2$. When Boolean networks are used to model biological systems, such as gene regulatory networks, the dynamics of these biological systems are modeled using the discrete dynamical systems defined by these Boolean functions via iterations. Given a Boolean function $f:F_2^n\rightarrow F_2^n$, we call the elements of $F_2^n$ the {\it states}, and call the function graph (directed) $S(f)$ of $f$ the {\it state space}. So the nodes of $S(f)$ are the elements of $F_2^n$, and there is a directed edge from node $a$ to node $b$, denoted by $a\longrightarrow b$ in $S(f)$, if and only if $f(a) = b$. We call a directed cycle of $S(f)$ a {\it limit cycle}. For a limit cycle of length $1$, we call the unique node that forms the limit cycle {\it a fixed point} or {\it a stable state}. 
\par
Ideally, one would like to be able to compute these limit cycles from a given description of the Boolean network, since the limit cycles provide essential information about a Boolean network and graphical diagrams can only be drawn for small $n$'s. However, even the detection of a stable state of a Boolean network is NP complete in general \cite{Tamu1, Zhan}. Thus, developing new approaches and algorithms to address the need of modeling complex systems have been the focus of research. We now have some of the tools capable of finding all stable states of Boolean networks with about $100$ nodes if each of the nodes has only a few connections \cite{ADAM}. On the other hand, detecting the limit cycles of length $> 1$ is usually based on the detection of stable states. Namely, one computes the stable states of the composition function $f\circ f$ for limit cycles of length $\le 2$, then compute the stable states of $f\circ f\circ f$ for limit cycles of length $\le 3$, etc., and then gets the limit cycles of various lengths by taking the differences. There have been studies devoted to linking the topology of the dependency graph of a Boolean network with information on its limit cycles \cite{Col1, Gro2, Jar1, Jar2, Jar4, Laub, Zou10}, but due to the nature of the problem, studies in this direction must deal with special types or particular features of Boolean networks \cite{Kauf1,Mol,Zou12}. Simulation methods, in which one runs iterations by starting with some randomly selected initial states in order to gain information about a Boolean network's long term behavior, are still the main tools used in analyzing biological Boolean models \cite{Albe, Kino09, SR07}, in particular large Boolean models.
\par
We propose to consider various iteration schemes for Boolean models as parameterized Boolean networks. From a parameterized Boolean network, we can construct sequences of functions that define the dynamics of different iteration schemes, and thus will allow us to analyze the long term behaviors of the models using existing analytical tools. We will consider parameterized Boolean networks under a more general setting by using vectors to create more than two levels of expressions for the nodes. We now give some motivations for our work.
\par  
Suppose that a Boolean network with $n$ nodes is given by 
\bea\label{ea1}
\mathbf{f}=(f_1,\ldots,f_n): F_2^n\rightarrow F_2^n,
\eea
where $f_1,\ldots,f_n$ are functions $F_2^n\rightarrow F_2$ that define the update rules for the nodes. Depending on the iteration schemes, the Boolean models derived from $f$ are named differently \cite{Ger, Kauf69, Shmu02}. The basic model, called the synchronous model, is based on the formulation that the states of the nodes at time $t+1$ are given by the value of $f$ at time $t$, i.e., all nodes are updated simultaneously by  
\bea\label{ea2}
x_i(t+1) = f_i(x_1(t),\ldots,x_n(t)), \quad 1\le i\le n.
\eea 
If the same function $f$ is used at each iteration step, but the nodes are not assumed to be updated simultaneously as in (\ref{ea2}), then we have the asynchronous models, which vary depending on their updating schemes. 
\par
For example, the so-called deterministic moduli asynchronous random Boolean networks \cite{Ger} have two parameters $P_i$ and $Q_i$ for each node $i$, where $Q_i<P_i$ are positive integers generated randomly and remain fixed throughout the iteration process. Node $i$ will be updated when $t\equiv Q_i\; \mbox{mod $P_i$}$ (i.e. the remainder of $t$ divided by $P_i$ is $Q_i$). If, at a certain time step, there are several nodes fulfill their updating conditions, then these nodes will be updated from left to right, with any to be updated node taking into account of the states of the already updated nodes.
\par
For another example \cite{Cha}, one uses a permutation $p=(p_1,p_2,\ldots,p_n)$ of $(1,2,\ldots, n)$, randomly generated at each time step, to decide the updating order for the nodes. Node $i$ is the $p_i$-th node to be updated according to 
\bea\label{ea3}
x_i(t+1) = f_i(x_1(t_{p_1}),x_2(t_{p_2}),\ldots,x_n(t_{p_n})),
\eea
where $t_{p_1} =t$ if node $1$ should be updated later than node $i$ (i.e. $p_i<p_1$), and $t_{p_1} = t+1$ if node $1$ has been updated already (i.e. $p_1<p_i$). Other $t_{p_j}$'s are defined similarly. 
\par 
These two examples have something in common: the updating schemes work at one node at each iteration step, the new state of the to be updated node is computed based on the current states of all nodes, and the entire network is considered as updated when a round of updates for all nodes is completed. For other asynchronous update schemes depending on various modeling needs, we refer the readers to the references cited. These different Boolean models have a fixed updating rule $f$, they only differ by the updating schemes. Since these asynchronous schemes allow more complexities to be built into a model to address the needs in applications, it is desirable to investigate how different updating schemes affect the dynamics of $f$ as a synchronous model, and to use this information to understand the more complicated asynchronous models better.
\par
 Study based on simulations has revealed some interesting properties of the effects of updating schemes on the dynamics of $f$ \cite{Ger}, but our understanding of these modeling schemes is still rather poor in comparison with the synchronous models. To analyze these models, it is necessary to put these descriptive methods into some clearly defined mathematical formulations. One way to achieve this is to consider these different updating schemes by starting with parameterized Boolean networks. For our purpose here, a parameterized Boolean network is a set of Boolean functions, and, since we are dealing with finite sets, such a Boolean network can be parameterized by a sequence of integers (see Section 3 for a detailed definition). This setting will facilitate the analysis of the long term behaviors these Boolean models. 
 \par
 We remark that the idea of using parameters for Boolean networks is not new. The so-called random Boolean networks \cite{Kauf69} and probabilistic Boolean networks \cite{Shmu02} can be considered as certain types of parametrized Boolean networks. However, the fact that a probabilistic type Boolean network does not start with a fixed Boolean network creates extra complexities for the model, which further limits the size of a network we can analyze \cite{Ivan09}. We will not consider probabilistic Boolean networks here. 
\par
There are different methods for handling systems where each node has expressions not restricted to just ``on'' or ``off''. Among these, one approach is to use finite fields with more than two elements, and another is to use vector values for the expressions of each node in Boolean networks  \cite{Ger, ADAM, Hin, Laub, SLK99}. For multi-expression networks, and, in particular, large size networks, using Boolean networks, if possible, can still provide computational advantage, since many high performance algorithms exist for Boolean logic problems such as SAT (see \cite{Tamu1, Zhan}, and the references therein, see also Theorem 3.1 below). Methods of allowing multi-dimensional vector values for each node using Boolean networks exist in the literature. For example, \cite{Ger} suggests clustering several nodes together to form a single node. In \cite{Hin}, a similar approach was applied to developing a method of using Boolean networks to capture the basic properties of biological systems described by delay differential equations. Instead of using a single variable $x$ to represent a gene, to capture the delay, $x$, $x_1$, and $x_{old}$ were used. However, though this approach is natural, a review of the literature indicates that some promotion is needed  for its application, such as in the increasing important study of biological systems which exhibit multi-state properties \cite{Lim}. This is due, perhaps, to the fact that these methods are hiding in the literature. In this paper, we consider parameterized Boolean networks that allow more than two levels of expression and allow {\it different numbers of expression levels for different nodes} by assigning vector values to the nodes.  For example, control types of nodes in a network can just have ``on'' or ``off'' as their states, while other nodes can have more than two states, such as ``high'', ``partial high'', ``partial low'', and ``low''.
\par 
\medskip
We organize this paper as follows. In section 2 and section 3, we define multi-expression Boolean networks and parametrized multi-expression Boolean networks. In section 4, we consider the dynamics of the Boolean networks defined in section 2 and section 3. Several examples are provided in these sections to explain the concepts. In section 5, we consider the dynamics of certain random asynchronous Boolean models and give a theorem about the dynamics of these models based on stochastic matrix theory. In section 6, we prove a theorem which links the attractor sets of a deterministic moduli asynchronous random Boolean network to the limit cycles of a synchronous Boolean network. In section 7, we conclude this paper.
\section{Multi-Expression Boolean Networks}
\par
\medskip
We define Boolean networks which can capture different multi-expression levels in this section. In the next section, we will consider the parametrization of these Boolean networks. We will use Boolean polynomials to represent the Boolean functions.  Boolean polynomials are elements of the Boolean ring 
\bea\label{eb1}
F_2[x_1,\cdots, x_M]/(x_i^2-x_i, 1\le i\le M),
\eea
i.e., the quotient of the polynomial ring $F_2[x_1,\cdots,x_M]$ by the ideal generated by $x_i^2-x_i, 1\le i\le M$. The conversion between logical expressions and Boolean polynomials is given by 
\begin{gather*}
x_i\wedge x_j = x_ix_j,\;\;
x_i\vee x_j = x_i+x_j+x_ix_j,\\
\neg x_i = x_i+1.
\end{gather*}
\par
We use the following notation: 
\be
\mathbf{x}=(x_1,\ldots, x_M), \quad \mathbf{x}(t)=(x_1(t),\ldots, x_M(t)),\\
\mathbf{x}^{a} = x_1^{a_1}\cdots x_M^{a_M}\;\mbox{if $a=(a_1,\ldots, a_M)$},
\ee  
and write $B[\mathbf{x}]$ for the Boolean polynomial ring of (\ref{eb1}).
\par
\begin{definition} A Multi-Expression Boolean Network (MEBN) with $n$ nodes is a function
\bea\label{eb2}
f=(f_1, f_2,\ldots, f_n): F_2^M\longrightarrow F_2^M,
\eea
where 
\be
f_i: F_2^M\longrightarrow F_2^{M_i},\;\;1\le i\le n,
\ee
for some positive integers $M_i$ such that $M_1+\cdots+M_n = M$.  We denote the nodes by $V_i,\;1\le i\le n$.
\end{definition}
\par
This definition of an MEBN allows the possibility of different levels of expressions for different nodes. If all $M_i=1$, then $f$ is just a Boolean network where each node has only two possible states ``on'' and ``off'' as usual. In this case, we simply use the variables $x_i$ for the nodes. The limit cycles of an MEBN are defined as before. For example, the fixed points are defined by the equation $f(\mathbf{x}) = \mathbf{x}$. We elucidate our MEBN definition by considering some examples.
\par
\begin{example} Consider a system with two nodes. The first node can have $4$ possible states: ``high'' ($\sim 100\%$), ``partial high'' ($\sim 75\%$), ``partial low'' ($\sim 25\%$), and ``low'' ($\sim 0\%$). The second node has only two states: ``on'' and ``off''. A Boolean network used to describe this system can then be given by $f = (f_1,f_2)$, where $f_1 : F_2^3\rightarrow F_2^2$ and $f_2 : F_2^3\rightarrow F_2$. The possible states of the first node are given by the vectors in $F_2^2 = \{(00), (01), (10), (11)\}$, with $(11)$ meaning ``high'', $(10)$ meaning ``partial high'', $(01)$ meaning ``partia low'', and $(00)$ meaning ``low''. In applications, we can use the second node to control the first node. For example, the second node being ``off'' causes the first node's expression level to decrease, and the second node being ``on'' causes the first node's expression level to increase. Here is such a network:
{\small \be
f_1 &=& (x_1x_2 + x_1x_3 +x_2x_3, x_1x_2 + x_2x_3 + x_1x_3 + x_1 + x_3),\\
 f_2 &=& x_3.
\ee}
This Boolean network has two trajectories and two fixed points. The second node serves as a ``switch'' between these two trajectories. For instance, if the system is at the stable state $(11),1$ and the second node changes its status from $1$ to $0$, then the system changes its state to $(11),0$. If the second node remains at the ``off'' position, then the system will evolve down to the stable state $(00),0$ and remain there until the second node changes its status again.
The state space graph of this Boolean network is given by Fig. 1.
\par
\begin{figure}[h]
\begin{center}
\includegraphics[width=1in, height=1.2in]{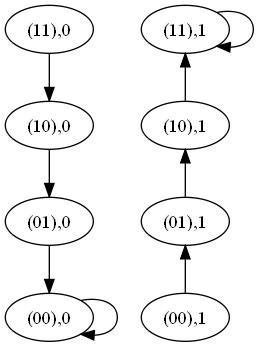} 
\caption{{\footnotesize The second node, whose states are given by the last digit, serves as a switch between the two trajectories.}} 
\label{Fig1}
\end{center}
\end{figure}
\par
\end{example}
\par
For expression levels other than $2^k$, one can choose the least $k$ such that $2^k$ upper bounds the number of levels, then considers grouping several vectors together and treats all of them as for the same expression level.
\par
\begin{example} Suppose an DNA replication switch shows $3$ states: ``on'', ``intermediate'', and ``off''. Then we can use 
\be
F_2^2= \{(00), (01), (10), (11)\}
\ee
for this switch, and treat both $(01)$ and $(10)$ as intermediate states when perform analysis of the system using a Boolean model. For $5$ levels, we can use $F_2^3$. For example, we can treat $(111)$ and $(110)$  both as ``on'', $(101)$ and $(100)$ both as ``high partial'', $(011)$ as ``medium'', $(010)$ and $(001)$ as ``low partial'', and $(000)$ as ``off''. Of course, these interpretations depend on the actual modeling needs. Here is a concrete example of an MEBN, in which the first node has five levels of expressions and the second node serves as a control switch assuming two states ``on'' and ``off'':
{\small
\be
f_1 &=& (x_2x_3x_4 + x_1x_4x_2 + x_1x_4x_3 + x_1x_4 + x_1x_3x_2+ x_1x_3 + x_1x_2,\\
  {}  &{}& x_1x_3x_2 + x_1x_4x_3 + x_1x_4x_2 + x_2x_3x_4 + x_1x_2+ x_2x_4\\ 
      {}  &{}&    + x_1 + x_2x_3 + x_3x_4 + x_1x_4 + x_1x_3,\\
  {}  &{}&  x_1x_2x_3x_4 + x_2x_3x_1 +  x_1x_4x_3 + x_1x_4x_2 + x_2x_3x_4\\
  {}  &{}& + x_1x_2 + x_2x_4+ x_1 + x_2 + x_4 + x_2x_3 + x_3x_4 + x_1x_4 + x_1x_3)\\
f_2 &=& x_4
\ee
}
The state space graph of this Boolean network is given by Fig. 2.
\par
\begin{figure}[h]
\begin{center}
\includegraphics[width=3.5in, height=1.5in]{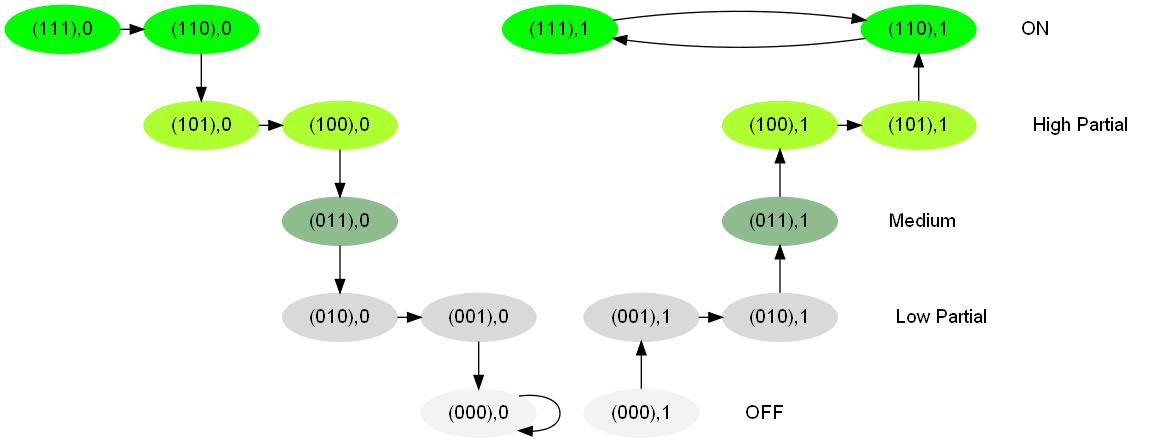} 
\caption{{\footnotesize The nodes with the same colors represent the same expression levels. There are $5$ levels. Green stands for ``on'', greenyellow for ``high partial'', seagreen for ``medium'', gray for ``low partial'', and light for ``off''. When the second node is ``on'' (``off''), the expression level increases (decreases) until stabilized. The space graph has one stable state $0000$ and a limit cycle of length $2$ given by $\{1111, 1101\}$.}} 
\label{Fig2}
\end{center}
\end{figure}
\end{example}
\par\medskip
It is not hard to see that, in general, the only limit cycles which always exist and remain the same for the iteration schemes mentioned in the introduction section are the fixed points. In general, asynchronous schemes can destroy the limit cycles of length $\ge 2$ of a Boolean network $f$ \cite{Ger} (see examples in section 4). Even for the probabilistic Boolean networks, the transition probabilities are computed differently according to whether or not the current state is a fixed point \cite{Ivan09}. Thus as far as modeling is concerned, regardless of the iteration scheme used, improving our ability to detect the fixed points of a Boolean network is fundamental in dealing with complex biological systems. We describe a theorem for this purpose.   
\par
For an MEBN $f = (f_1,\ldots,f_n)$, where 
\be
f_i = (f_{i1},\ldots, f_{iM_i}) : F_2^M\longrightarrow F_2^{M_i},\;\;1\le i\le n,
\ee
the equation $f(\mathbf{x})=\mathbf{x}$ that defines the fixed points of $f$ is given explicitly as
\be
f_{ij}(\mathbf{x}) = x_{ij}, \; 1\le i\le n,\; 1\le j\le M_i,
\ee
where we identify $(x_{11},\ldots, x_{1M_1},x_{21},\ldots,x_{2M_2},\ldots,x_{n1},\ldots,x_{nM_n})$ with $\mathbf{x} = (x_1,x_2,\ldots,x_M)$.  We define the following Boolean function associated with $f$
\bea
m_{f} = 1+\prod_{ \stackrel
{1\le i\le n} {1\le j\le M_i}}(f_{ij}+x_{ij}+1)\;:\;F_2^M\longrightarrow F_2.
\eea
Then we have the following theorem \cite{Zou12}:
\begin{theorem} The fixed points of $f$ are the solutions of the equation $m_f = 0$. 
\end{theorem}
\par\medskip
{\bf Remark.} The above theorem shows that detecting a single fixed point of a Boolean network is equivalent to the SAT problem and hence NP complete in general.
\par\medskip
From the above theorem, we have the following corollary.
\begin{corollary} A Boolean network $f$ does not have a fixed point if and only if $m_f=1$.
\end{corollary}
\par\medskip
\begin{example} Consider the MEBN $f$ of Example 2.1. We have
\be
m_f =x_1x_2+x_1x_3+ x_2x_3+x_1+x_2+x_3.
\ee
The solutions of $m_f=0$ are $(000)$ (set $x_1=0$ and then solve the resulted equation) and $(111)$ as given in Fig. 1.
\end{example}
\par
For certain applications of Boolean networks, especially when asynchronous schemes are applied to model biological systems, it is necessary to consider a more general definition of ``limit cycles'' for the long term behavior of a network. We will address this point in section 4.
\section{Parametrized Boolean Networks}
\par
There are different possible ways to introduce parameters for Boolean networks. For our purposes, we make the following definition. 
\par
\begin{definition} Given positive integers $M_i,\;1\le i\le n$, and $M$, such that $M_1+\cdots+M_n = M$, we define a Parametrized Multi-Expression Boolean Network (PMEBN) to be a set of functions $S=\{f,g,h,\ldots\}$, where each function is an MEBN as defined by Definition 2.1.
\end{definition}
\par
Since for any positive integer $M$, there are only finitely many (the number is $(2^M)^{2^M}=2^{M2^M}$) functions $F_2^M\rightarrow F_2^M$, we can always use a set of positive integers $\{1,2,\ldots,s\}$ to parametrize the elements of a PMEBN. If a PMEBN contains only one function, then it is just an MEBN as in section 2.
\begin{example}
Suppose that we want to use a Boolean network to model a gene regulatory network which the interactions among the genes are only partially known. Then we may want to use a parametrized Boolean network, i.e., instead of picking a single one that fits the known information, we can describe a set of Boolean networks, all fit the known information, and study this set of Boolean networks as a whole to gain information about the gene regulatory network. 
\end{example}
\par
The situations described in Example 3.1 are typical in real world applications. We give some more detail for a special case when the Boolean network we are looking for has known values on a given subset of the state space. The problem is this: we are interested in a Boolean network $f : F_2^M\rightarrow F_2^M$, whose values on a given subset $D\subset F_2^M$, small in size compare with the whole state space $F_2^M$, are known. Because the information on $f$ is partial, there are many Boolean networks fit the profile. Observe the following:
\par
(a) If $f$ fits the profile of the Boolean network we are looking for and $h$ is a Boolean network that is identically $0$ on $D$, then $f+h$ also fits the profile.
\par
(b)  If two Boolean networks $f$ and $g$ take identical values on the states in $D$, then $f-g = f+g = (f_1+g_1,\ldots, f_M+g_M)$ (since this is over $F_2$) is identically $0$ ($=(0\cdots 0)$) on $D$.
\par
(c) If $h$ is identically $0$ on $D$ and $q : F_2^M\rightarrow F_2$ is arbitrary, then $gh=(gh_1,\ldots, gh_M)$ is also identically $0$ on $D$.
\par
So if we let the set of Boolean networks that are identically $0$ on $D$ be $S$, then the Boolean networks that fit the profile are given by the set $f+S=\{f+g\;|\; g\in S\}$, where $f$ is a fixed Boolean network that fits the profile.  According to our definition, $f+S$ is a PMEBN. One of the properties that distinguish Boolean networks from polynomial systems over other fields is that the set $S$ can be described explicitly by the following theorem \cite{Zou11,Zou12}:
\begin{theorem} There exists a unique $q : F_2^M\rightarrow F_2$ such that $g=(g_1,\ldots, g_M)\in S$ if and only if $g_i = m_iq$ for some $m_i : F_2^M\rightarrow F_2$, $1\le i\le M$. Furthermore, $q$ can be computed explicitly by 
\bea\label{ec1}
q = 1+\sum_{a\in D}p_a,
\eea
where $p_a = \prod_{i=1}^M(x_i+a_i+1)$ for $a=(a_1,\ldots,a_M)\in F_2^M$.
\end{theorem}
\par
A special $f=(f_1,f_2,\ldots,f_M)$ that fits the profile can also be found explicitly by using the following formula to construct its component functions:
\bea\label{ec2}
f_i = \sum_{\substack{a\in D\\f_i(a)=1}}p_a,\quad 1\le i\le M.
\eea 
\par
These closed formulas exist only in the Boolean case. For other finite fields, a computation of the generators, such as a Gr\"{o}bner basis, would be needed. Here is an example.
\begin{example} Suppose we know that a network $f = (f_1,f_2) : F_2^3 \rightarrow F_2^3$, where $f_1 : F_2^3 \rightarrow F_2^2$ and $f_2 : F_2^3 \rightarrow F_2$, satisfies $f(000)=(000)$ and $f(101)=(111)=f(111)$ (see Example 2.1). It is easy to see that one of the networks satisfying these conditions is $h=(h_1,h_2)$, where $h_1=(x_1,x_1)$ and $h_2=x_1$. Here the set $D =\{(000),(101),(111)\}$ and 
\be
q &=& 1 + p_{000} + p_{101} + p_{111}\\
{} &=& 1+ (x_1+1)(x_2+1)(x_3+1) +  x_1(x_2+1)x_3 +  x_1x_2x_3\\
{} &=& x_1x_2x_3+x_1x_2+x_2x_3+x_1+x_2+x_3.
\ee
Compare with the function $f=((f_{11},f_{12}),f_2)$ of Example 2.1, which also satisfies the given conditions,  we have
{\small
\be
f_{11} &=& x_1+(x_1x_2+x_1x_3+x_2x_3+x_1)q,\\
f_{12} &=& x_1+(x_1x_2+x_1x_3+x_2x_3+x_3)q,\\
f_2 &=&  x_1+ (x_1+x_3)q.
\ee }
That is 
{\small
\be
f = h + q(x_1x_2+x_1x_3+x_2x_3+x_1,x_1x_2+x_1x_3+x_2x_3+x_3,x_1+x_3).
\ee}
\end{example}
We give another example of PMEBNs.
\begin{example} Given an MEBN $f =(f_1,\ldots, f_n)$ as in Definition 2.1, we can construct a PMEBN $P_f$ by setting
$P_f=\{p_{f_1},\ldots, p_{f_n}\}$, where the functions $p_{f_i} : F_2^M\rightarrow F_2^M$ are defined by
\be
p_{f_1}(\mathbf{x}_1,\mathbf{x}_2,\ldots, \mathbf{x}_n) &=&  (f_1(\mathbf{x}),\mathbf{x}_2,\ldots, \mathbf{x}_n), \\
p_{f_2}(\mathbf{x}_1,\mathbf{x}_2,\ldots, \mathbf{x}_n) &=&  (\mathbf{x}_1,f_2(\mathbf{x}),\ldots, \mathbf{x}_n), \\ 
{} &\vdots& {}\\
p_{f_n}(\mathbf{x}_1,\mathbf{x}_2,\ldots, \mathbf{x}_n) &=& (\mathbf{x}_1,\mathbf{x}_2,\ldots, f_n(\mathbf{x})),
\ee
where $\mathbf{x} = (\mathbf{x}_1,\mathbf{x}_2,\ldots, \mathbf{x}_n)$.
\end{example}
We will see in the next section that the asynchronous Boolean models considered in this paper can be analyzed as dynamical systems of PMEBNs under the setting of Example 3.3.
\par\medskip
\section{Dynamics of Parametrized Boolean Networks}
\par
We consider dynamical systems associated with a PMEBN in this section.
\par
\begin{definition} Given a PMEBN $P$ and a sequence of functions $s_P=\{f_i\}_{i=1}^{\infty}$, where $f_i\in P$ are chosen according to some rule $s$, we define the corresponding dynamical system of $s_P$ by 
\bea\label{ed1}
\qquad\mathbf{x},\;f_1(\mathbf{x}),\; (f_2\circ f_1)(\mathbf{x}),\;\ldots,\; (f_t\circ\cdots\circ f_1)(\mathbf{x}),\;\ldots,\;\forall\;\mathbf{x}\in F_2^M.
\eea
\end{definition} 
\par
We consider some Boolean models we have seen according to this definition.
\par
\begin{example} {\bf Synchronous Boolean Models.} If $P=\{f\}$, there is only one sequence that can be formed from $P$, that is, all $f_i=f$. In this case, the corresponding dynamical system is just the usual synchronous model defined by $f$, where the $t$th iteration is 
\be
f^t(\mathbf{x})=(\underbrace{f\circ f\circ\cdots\circ f}_{\mbox{$t$ terms}})(\mathbf{x}), \quad t=1,2,\ldots.
\ee
\end{example}
\par
The two asynchronous Boolean models mentioned in the introduction section can be treated as dynamical systems of PMEBNs.
\begin{example} {\bf Random Asynchronous Boolean Models.} 
For a given MEBN $f$, let $P_f=\{p_{f_1},\ldots, p_{f_n}\}$ be the PMEBN defined in Example 3.3. We define a sequence, denoted by $s_f^r$, as follows. The sequence consists of functions $g_i\in P_f,\; i\ge 1$, such that each block of $n$ functions 
\bea\label{ed2}
g_{tn+1},\; g_{tn+2},\;\ldots,\; g_{(t+1)n},\; t = 0, 1, 2, \ldots,
\eea
is a random permutation of the $n$ functions $(p_{f_1},\ldots, p_{f_n})$. Let the dynamical system be defined as in Definition 4.1. For any $\mathbf{x}\in F_2^M$, let 
\be
\mathbf{x}_0 = \mathbf{x},\; \mathbf{x}_1=g_1(\mathbf{x}),\;\ldots,\;\mathbf{x}_m=g_m(\mathbf{x}_{m-1}),\;\ldots,
\ee 
and write $(fr)_m(\mathbf{x})=\mathbf{x}_m$. Then the dynamics of the random asynchronous Boolean model defined by $f$ can be studied by using the subsequence $\{(fr)_{tn}(\mathbf{x})\}_{t=0}^{\infty}$.   
\end{example}
\par
Note that we do not use the full sequence $\{(fr)_m(\mathbf{x})\}_{m=0}^{\infty}$ to study the long term behavior of a random asynchronous Boolean model, since the whole network is considered as updated only after a full round of updates is completed, and according to the definition, each $g_m$ updates only one node. Note that the time steps are indexed by $t$.
\par
\begin{example} {\bf Deterministic Moduli Asynchronous Boolean Models.} Recall that a deterministic moduli asynchronous Boolean model is defined by a fixed $f$ and pairs of positive integers $Q_i<P_i,\;1\le i\le n$. Let $id$ be the identity function (i.e. $id(\mathbf{x})=\mathbf{x},\;\forall \mathbf{x}$), and let $P'_f=\{id,p_{f_1},\ldots, p_{f_n}\}$. We define a sequence $s_f^d$ as follows. The sequence consists of functions $g_i\in P'_f,\; i\ge 1$, such that each block of $n$ functions 
\bea\label{ed3}
g_{tn+1},\; g_{tn+2},\;\ldots,\; g_{(t+1)n},\; t = 0, 1, 2, \ldots,
\eea
are defined by $g_{tn+i}=p_{f_i}$ if $t+1\equiv Q_i$ (mod $P_i$), or $g_{tn+i}=id$ otherwise. As in Example 4.2, for any $\mathbf{x}\in F_2^M$ and $m\ge 1$, let 
\be
\mathbf{x}_0 = \mathbf{x},\; \mathbf{x}_1=g_1(\mathbf{x}),\;\ldots,\;\mathbf{x}_m=g_m(\mathbf{x}_{m-1}),\;\ldots,
\ee 
and write $(fd)_m(\mathbf{x})=\mathbf{x}_m$. Then the dynamics of the deterministic moduli asynchronous Boolean model defined by $f$ can be studied by using the subsequence $\{(fd)_{tn}(\mathbf{x})\}_{t=0}^{\infty}$. 
\end{example}
\par\medskip
We consider a more general concept than the concept of limit cycles for the long term behavior of a PMEBN.
\begin{definition} Given a PMEBN $f$ and an associated sequence of functions $g_m(\mathbf{x}),\;m=1,2,\ldots$, that defines the dynamical system $h_m=g_1\circ\cdots\circ g_m$. Suppose the dynamics of the given system is defined by the subsequence $h_{k_i}(\mathbf{x}),\;0=k_0<k_1<k_2<\cdots$. A subset $S\subset F_2^M$ is called an attractor set of the dynamical system if 
\par
(i) $h_{k_i}(S) \subset S,\;\forall i$, and 
\par
(ii) for any $\mathbf{y}\in S$ and any positive integer $N$, there exists $N<k_i$ such that $\mathbf{y}\in h_{k_i}(S)$. 
\par
An attractor set $S$ is called stable, if there is a positive integer $N$ such that for all $N<i$, $h_{k_i}(S)=S$.  An attractor set $S$ is called indecomposable if there is a fixed $\mathbf{x}\in F_2^M$ such that for any $\mathbf{y}\in S$ and any positive integer $N$, there exists $i>N$ such that $h_{k_i}(\mathbf{x})=\mathbf{y}$.
\end{definition}
\par\medskip
We give some examples for this definition.
\begin{example} For any Boolean network $f$, any limit cycle of the synchronous Boolean model defined by $f$ is an indecomposable stable attractor set, and any union of limit cycles is a stable attractor set. For asynchronous models defined by $f$, the stable states are indecomposable attractor sets, each contains a single state.
\end{example}
\par
Consider a concrete example. 
\begin{example} Let $f=(f_1,f_2,f_3)$, where 
\be
f_1 = x_2,\quad f_2=x_3+1,\quad f_3=x_2+x_3.
\ee
The state space graph of this Boolean network is given by Figure 3. 
\begin{figure}[h]
\begin{center}
\includegraphics[width=1.5in, height=1in]{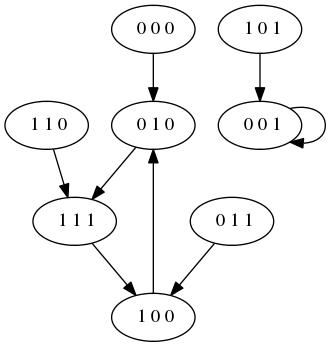} 
\caption{{\footnotesize Sate space graph of $f$.}}
\label{Fig3}
\end{center}
\end{figure}
\par
For this Boolean network, no matter what scheme we use for the model, the system will end up going into the stable state $(001)$ from the initial state $(101)$. This is because the updates of the second and the third nodes do not change their statuses, whereas an update of the first node will send the network from $(101)$ to $(001)$. However, things are quite different for the branch with a limit cycle of length $3$. Suppose we use the random asynchronous scheme for the Boolean model given by $f$. Then an update may send the system from this branch to the other branch. For instance, starting from the state $(011)$, if the second node is updated first, then the system will get into the stable state $(001)$. Consider the three states that form the limit cycle of length $3$. We can see that, from $(100)$, the next state of the system remains in the branch no matter how the updates are carried out. But $(111)$ is special. From $(111)$, if the second node is updated first, the system will go to $(101)$ and will end up at the stable state $(001)$. Similarly, from $(010)$, there is a $\frac{1}{6}$ probability that the system will end up in the other branch (the update order is $3,2,1$). We can say that by cycling around the limit cycle of length $3$, the probability that the system stays at its current branch is $<\frac{2}{3}$, and the probability of moving to the other branch is $>\frac{1}{3}$. Since the probability of staying at the current branch for cycling around the limit cycle $t$-times is $<(\frac{2}{3})^t$, then as $t$ gets large, the chance that the system stays in its current branch approaches $0$. In fact, for the $5$th round, the probability reduces to about $0.13$. Thus, according to Definition 4.2, as a random model, this Boolean model has only one attractor set, which is given by the single stable state $(001)$.    
\end{example}
\par
We will consider random models in more detail in section 5. Consider another Boolean network with $3$ nodes. 
\begin{example} Let $f=(f_1,f_2,f_3)$, where 
{\small\be
f_1= x_1,\quad
f_2=x_1x_3 + x_1 + x_3,\quad
f_3=x_1x_2x_3 + x_1x_3 + x_1 + x_2 + 1.
\ee}
The state space graph of this Boolean network is given by Fig. 4. 
\begin{figure}[h]
\begin{center}
\includegraphics[width=1.5in, height=1in]{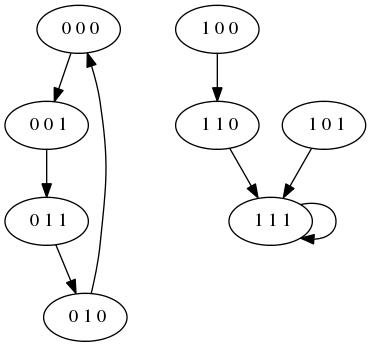} 
\caption{{\footnotesize Sate space graph of $f$.}}
\label{Fig4}
\end{center}
\end{figure}
For this Boolean network, the fixed point and the limit cycle of length $4$ are simple attractor sets regardless of the iteration schemes used.
\par
\end{example}
\par
We have a sufficient condition for a limit cycle of a MEBN $f$ to be an attractor set for different iteration schemes.
\begin{proposition} For a given MEBN $f$, let $P_f=\{p_{f_1},\ldots, p_{f_n}\}$ be the PMEBN defined in Example 3.3, and let $L$ be a limit cycle of $f$ (as a synchronous model). If $p_{f_i}(L) = L,\; 1\le i\le n$, then $L$ is a stable attractor set for all iteration schemes considered in this paper (i.e. no flip function is allowed). 
\end{proposition}
\begin{proof}
This is self evident. Note however, a limit cycle can decompose into indecomposable attractor subsets.
\end{proof} 
\par\medskip
\section{Dynamics of Random Asynchronous Boolean Models}
\par
We consider the dynamics of random asynchronous Boolean models of Example 4.2 in this section. Let $f$, $P_f=\{p_{f_1},\ldots, p_{f_n}\}$, and $\{(fr)_m(\mathbf{x})\}_{m=0}^{\infty}$ be defined as in Example 4.2. Recall that the dynamic of such a system is defined by $\{(fr)_{tn}(\mathbf{x})\}_{t=0}^{\infty}$.
\par
For any pair of states $\mathbf{x},\mathbf{y}\in F_2^M$, we define $p_{\mathbf{x},\mathbf{y}}\in [0,1]$ as follows. Let $a_{\mathbf{x},\mathbf{y}}$ be the number of permutations $\{p'_{f_1},\ldots, p'_{f_n}\}$ of $P_f$ such that $p'_{f_n}\circ\cdots\circ p'_{f_1}(\mathbf{x})=\mathbf{y}$, and set $p_{\mathbf{x},\mathbf{y}}=a_{\mathbf{x},\mathbf{y}}/n!$. Consider the stochastic matrix $P_{fr}$ with entries $p_{\mathbf{x},\mathbf{y}}$. We can label the states such that $P_{fr}$ has the following block upper triangular form (see p. 695 of \cite{Mey}):
\bea\label{E1}
P_{fr}=\left(\begin{array}{c|c}
        \begin{array}{ccc}
          P_{11} & \cdots & P_{1r}\\
          {}     & \ddots & \vdots\\
          {}     & {}     &  P_{rr}
        \end{array} & 
        \begin{array}{ccc}
          P_{1,r+1} & \cdots & P_{1m}\\
          \vdots   & {} & \vdots\\
          P_{r,r+1} & \cdots & P_{rm}
        \end{array} \\ \hline
        {} &
        \begin{array}{ccc}
          P_{r+1,r+1} & {} & {}\\
          {}     & \ddots & {}\\
          {}     & {}     &  P_{mm}
        \end{array}  
      \end{array}
\right),
\eea
where each $P_{11},\ldots, P_{rr}$ is either irreducible or $(0)_{1\times 1}$, and $P_{r+1,r+1},\ldots, P_{mm}$ are irreducible. We may further assume that all $(1)_{1\times 1}$, if any (these correspond to the fixed points of $f$), are at the low right end conner. For each $1\le i\le m-r$, let $\pi^T_i$ be the left-hand Perron vector for $P_{r+i,r+i}$. We have the following theorem:
\begin{theorem} Given a PMEBN $f$, consider the random asynchronous Boolean model defined by $f$, and define the stochastic matrix $P_{fr}$ as in (\ref{E1}). 
\par
(i) For any $\mathbf{x}\in F_2^M$, there exists a positive integer $N_{\mathbf{x}}$ such that $(fr)_{tn}(\mathbf{x})$ belongs to one of the indecomposable attractor sets defined by the subsystems $P_{r+i,r+i},1\le i\le m-r$, for all $t>N_{\mathbf{x}}$.
\par
(ii) If $S_i=\{\mathbf{s}_{i1},\ldots,\mathbf{s}_{ic_i}\}$ is the attractor set that corresponds to $P_{r+i,r+i}$, $1\le i\le m-r$, then $c_i$ is equal to the size of $P_{r+i,r+i}$. If $(fr)_{tn}(\mathbf{x})$ is eventually in $S_i$, then as $t\rightarrow \infty$, the fraction of time that $(fr)_{tn}(\mathbf{x})=\mathbf{s}_{iu},1\le u\le c_i$, has the corresponding component of the vector $\pi^T_{i}$ as its limit.
\end{theorem}
\begin{proof} This follows directly from the corresponding results in \cite{Mey}.
\end{proof}
\par We consider an example.
\begin{example} Let $f=(f_1,f_2)$ be an MEBN, where $f_1=(f_{11},f_{12})$ and
\be
f_{11} &=& x_1x_2 +x_1x_3+x_2+x_3,\\
f_{12} &=& x_1x_2x_3 + x_1x_2 + x_2x_3 + x_1+ x_3,\\ 
f_2 &=& x_1x_2(x_3 + 1).
\ee
In this network, the first node has four expression levels and the second node has two expression levels. As a synchronous model, $f$ has $3$ limit cycles (see Fig. 5).
\begin{figure}[h]
\begin{center}
\includegraphics[width=1.5in, height=1in]{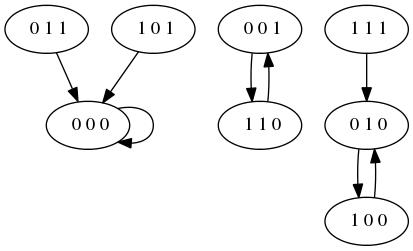} 
\caption{{\footnotesize Sate space graph of $f$.}}
\label{Fig5}
\end{center}
\end{figure}
If we consider the random asynchronous model defined by $f$ and order the states by
\be
111,\;	001,\;	110,\;	011,\;	101,\;	100,\;	010,\;	000,
\ee
then corresponding stochastic matrix
{\footnotesize
\be
P_{fr}=\left(\begin{array}{c|c|c}
        \begin{array}{ccccc}
          0 & 0.5 & {} & {} & {}\\
          {} & 0  & 0.5  & {} & {}\\
          {} & {} & 0 & 0.5  & {}\\
          {} & {} & {} & 0  & 0\\
          {} & {} & {} & {}  & 0
        \end{array} & 
        \begin{array}{cc}
          {} & {}\\
          {} & {}\\
          {} & {}\\
          0.5 & 0\\
           0 & 0.5
        \end{array} &
        \begin{array}{c} 
          0.5\\
          0.5\\
          0.5\\
          0.5\\
          0.5
        \end{array}\\ \hline
        {} &
        \begin{array}{cc}
          0 & 1\\
           1 & 0    
        \end{array}  & 
        \begin{array}{c}
                     0\\
                     0
        \end{array} \\ \hline
        {} & {} & 1
      \end{array}
\right).
\ee }
Thus the random asynchronous model has two simple attractor sets $\{000\}$ and $\{010,100\}$, the Perron vector for the latter is $(1/2,1/2)^T$, and the subsystem hitting each state in an equal amount of time in the long run. 
\end{example}
\par
{\bf Remark.} In applications, to reduce the complexity incurred by the number of permutations needed to be considered in the computation of the entries of the stochastic matrix in Theorem 5.1, we can employ the idea of MEBN by grouping the nodes which can be considered as updated synchronously together. If $f=(f_1,f_2,\ldots,f_M)$ and we can divide the nodes into $k$ ($<M$) groups such that the nodes in each group are updated simultaneously, then we only need to work with $k!$ permutations. 
\par\medskip
\section{Dynamics of Deterministic Moduli Asynchronous Boolean Models} 
\par
We consider attractor set structures of deterministic moduli asynchronous Boolean models in this section.  Given a fixed MEBN $f=(f_1,\ldots,f_n)$ and pairs of positive integers $Q_i<P_i,\;1\le i\le n$, let $P'_f=\{id,p_{f_1},\ldots, p_{f_n}\}$, $g_m$, and $(fd)_m$, $m=0,1,2,\ldots$, be defined as in Example 4.3. The model is {\it deterministic} in addition to the fact that $f$, $Q_i$, and $P_i$ are fixed (though they can be randomly generated), the updating order is also fixed: at a time step if several nodes are to be updated, then they are updated from left to right. 
\par
Recall that the long term behavior of the model, which will be denoted by $fd$, is defined by using the subsequence $\{(fd)_{tn}\}_{t=0}^{\infty}$. To simplify our writing, let $h_t=(fd)_{tn}$. Let $L$ be the least common multiple of $P_i,1\le i\le n$. Since each $P_i$ divides $L$, we have $m\equiv m+L\;(P_i), 1\le i\le n,\forall m$, and thus $g_{tn+i}=g_{(t+L)n+i}$. Therefore, the sequence of functions $\{g_m\}$ can be divided into identical blocks, each consists $Ln$ functions $g_1,g_2,\ldots,g_{Ln}$. Consider
\bea\label{F1}
h_L= (fd)_{Ln} =g_{Ln}\circ\cdots\circ g_2\circ g_1 : F_2^M\rightarrow F_2^M.
\eea
Then we have the following theorem.
\par
\begin{theorem} Notation as above.
\par
(i) The indecomposable attractor sets of $fd$ are given by the fixed points or cycles of length $>1$. The cycles of $fd$ of length $>1$, if any, are extensions of the limit cycles (including fixed points) of the synchronous Boolean model defined by $h_L$. The extension from a limit cycle of $h_L$ to an attractor cycle of $fd$ is as follows. If $\mathbf{x}$ and $\mathbf{y}$ are two states (could be identical) in a limit cycle of $h_L$ such that $h_L(\mathbf{x})=\mathbf{y}$, then $\mathbf{x}\rightarrow\mathbf{y}$ is extended by the distinct states contained in $\{h_1(\mathbf{x}),\ldots, h_{L-1}(\mathbf{x})\}$ naturally.
\par
(ii) Let the fixed point set of $f$ be $Z(f)$, let the fixed point set of $fd$ be $Z(fd)$, and let the fixed point set of each $h_i$ be $Z(h_i)$, $1\le i\le L$. Then 
\bea\label{F2}
Z(f)=Z(fd)=Z(h_1)\cap\cdots\cap Z(h_{L-1})\cap Z(h_L),
\eea 
and hence $Z(h_L)=Z(fd)$ if and only if $Z(h_L)\subset Z(h_1)\cap\cdots\cap Z(h_{L-1})$.
\end{theorem}
\begin{proof} (i) Let $S$ be an indecomposable attractor set of $fd$. Then $h_t(S)\subset S$ for any $t\ge 1$ by definition, so $h_L(S)\subset S$ and hence $S$ contains at least one limit cycle, say $C_L$, for the synchronous model defined by $h_L$. Let $\mathbf{x},\mathbf{y}\in C_L$ such that $h_L(\mathbf{x})=\mathbf{y}$, then 
\be
\mathbf{x}\rightarrow h_1(\mathbf{x})\rightarrow \cdots\rightarrow h_{L-1}(\mathbf{x})\rightarrow h_{L}(\mathbf{x}) =\mathbf{y}.
\ee
So we can extend the link $\mathbf{x}\rightarrow \mathbf{y}$ naturally. Do this for each pair of $\mathbf{x}\rightarrow \mathbf{y}$ in $C_L$, we then get a cycle which is an indecomposable attractor set of $fd$. This cycle must be equal to $S$. This proves (i).
\par
(ii) If $\mathbf{x}$ is a fixed point of $f$, then $g_m(\mathbf{x})=\mathbf{x}$ for all $m$, so $h_t(\mathbf{x})=\mathbf{x}$ for all $t$. If $\mathbf{x}$ is a fixed point of $fd$, then $h_t(\mathbf{x})=\mathbf{x}$ for all $t$, hence
\be
Z(f)\subset Z(fd)\subset Z(h_1)\cap\cdots\cap Z(h_{L-1})\cap Z(h_L).
\ee
On the other hand, we note that if $\mathbf{x}$ is a fixed point of $h_L$, then $\mathbf{x}$ is a fixed point of $fd$ if and only if $h_i(\mathbf{x})=\mathbf{x}$ for $1\le i\le L-1$ by part (i). We also note that if $\mathbf{x}$ is a fixed point of $fd$, then $\mathbf{x}$ is a fixed point of $h_t$ implies that $\mathbf{x}$ is fixed by each of the $f_i$ such that $g_{tn+i}=p_{f_i}$. Since for any $i$, $p_{f_i}$ must appear in $h_t$ such that $t+1\equiv Q_i$ (mod $P_i$), $\mathbf{x}$ is fixed by all $f_i$, and hence it is a fixed point for $f$. Thus we have 
\be
Z(f)\supset Z(fd)\supset Z(h_1)\cap\cdots\cap Z(h_{L-1})\cap Z(h_L),
\ee
and (\ref{F2}).
\end{proof} 
\par\medskip
{\bf Remark.} Theorem 6.1 gives a clear picture for the attractor sets of a deterministic moduli asynchronous Boolean model.  Additionally, the results address some of the questions from \cite{Ger} based on simulation observations, such as what types of attractor structure can such a Boolean model have. Note however, that the cycles of such a model are different from the cycles of a synchronous model: the system may not hit each state in a cycle in an equal amount of time in the long run since the states in $\{h_1(\mathbf{x}),\ldots, h_{L-1}(\mathbf{x})\}$ are not necessary distinct.
\par
\medskip
\begin{example} We use the Boolean network of Example 4.5 to construct an example of deterministic moduli asynchronous model. We have $f=(f_1,f_2,f_3)$, where 
\be
f_1 = x_2,\quad f_2=x_3+1,\quad f_3=x_2+x_3.
\ee
Assume updating scheme is defined by $P_i=2, 1\le i\le 3$, and $Q_1=Q_3 = 1, Q_2=0$. That is, nodes $1$ and $3$ are updated at the odd time steps, and node $2$ are updated at the even time steps. 
Then $L=2$, $Ln=6$, and  
\be
g_1=P_{f_1}, \; g_2=id,\; g_3=P_{f_3},\; g_4=id,\; g_5=P_{f_2},\;g_6=id.
\ee
So
\be
h_1 &=& P_{f_3}\circ P_{f_1} = (x_2,x_2,x_2+x_3),\\
h_2 &=& P_{f_2}\circ P_{f_3}\circ P_{f_1} = (x_2,x_2+x_3+1,x_2+x_3).
\ee
\begin{figure}[h]
\begin{center}
\includegraphics[width=1.5in, height=1in]{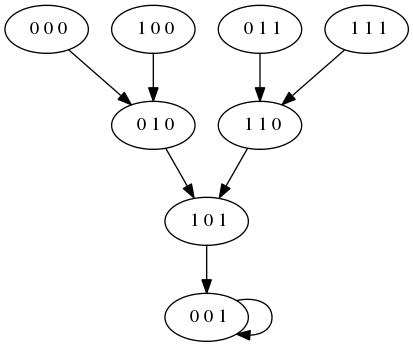} 
\caption{{\footnotesize Sate space graph of $h_2$.}}
\label{Fig6}
\end{center}
\end{figure}
Since the synchronous model defined by $h_2$ has only one limit cycle given by the single fixed point $(001)$ (see Fig. 6) and it is fixed by $h_1$, the deterministic moduli asynchronous model has only one attractor set given by the fixed point $(001)$. This is consistent with the observation in Example 4.5 as one would expect.
\end{example}
\section{Concluding remark}
\par
We have discussed two practical issues for the applications of Boolean networks in modeling biological systems: how to construct Boolean network based models in which the expression levels of the nodes can vary and are allowed to be more than two, and how to introduce parameters for Boolean models. Our treatment is algebraic and is based on polynomial systems over the field of two elements. One of our motivations is to extend the formulation of typical Boolean networks so we can analyze different asynchronous Boolean networks \cite{Albe, Ger} using algebraic methods \cite{Hin, Jar1, Laub, Mey, Zou11, Zou12}. We have applied our approach to perform analysis for two special types of asynchronous Boolean models, the random asynchronous model and the deterministic asynchronous moduli model. It is possible to combine the treatments of these two types of asynchronous Boolean models to analyze other asynchronous Boolean models, such as various asynchronous Boolean models described in \cite{Ger}, with some more complexity added to the notation and the descriptions of the associated dynamical systems. We will leave these analyses to a different report.
\par
We remark that algorithms for the detection of the fixed points of a Boolean network can be developed based on Theorem 2.1. Note that (compare with the remark after Theorem 2.1) SAT concerns whether a solution exists, not how to find all solutions. The author has developed an algorithm based on Theorem 2.1 for the computation of the fixed points of a Boolean network. Testing shows that it can handle interesting sizes Boolean networks, such as those having hundreds of nodes with an average connection $> 2$ in their dependence graphs. In particular, this algorithm works well with community-like networks where each (not too large) community can be densely connected, whereas communities are sparsely connected. The result will be reported somewhere else. We also remark that, though Theorem 5.1 provides a theoretical description for the long term behaviors of a random asynchronous PMEBN, the computational complexity limits the sizes of these random Boolean networks for which it can apply (see also the remark after Example 5.1). 


\begin{thebibliography}{9999}
\bibitem{Albe} Albert, I. {\it et al.} (2008), Boolean network simulations for life scientists, Source Code for Biology and Medicine 2008, 3:16.
\bibitem{Cha} Chaves M. and Albert, R. (2008), Studying the effect of cell division on expression patterns of the segment polarity genes, J. R. Soc. Interface 5, S71-S84.
\bibitem{Col1} Col\'{o}n-Reyes, O., Laubenbacher, R., and Pareigis, B. (2004), Boolean monomial dynamical systems, Annals of Combinatorics 8, 425-439.
\bibitem{Ger} Gershenson, C. (2004), Introduction to random Boolean networks, available online, 	arXiv:nlin/0408006v3
\bibitem{Gro2} Groote, J. F. and Kein\"{a}nen, M. (2005), A sub-quadratic algorithm for conjunctive and disjunctive BESs, Theoretical aspects of computing---ICTAC 2005, 532-545, Lecture Notes in Comput. Sci., 3722, Springer, Berlin.
\bibitem{ADAM} Hinkelmann, F. {\it et al.} (2011), ADAM: Analysis of discrete models of biological systems using computer algebra, BMC Bioinformatics 2011, {\bf 12}:295. 
\bibitem{Hin} Hinkelmann, F. and Laubenbacher, R. (2011), Boolean Models of Bistable Biological Systems, Discrete and Continuous Dynamical Systems Series S,  6(4), 1443-1456.
\bibitem{Ivan09} Ivanov, I. (2009), Boolean models of genomic regulatory networks:  Reduction mappings, inference, and external control, Current Genomics, 10, 375-387 
\bibitem{Jar1} Jarrah, A., Laubenbacher, R., Stigler, B., and Stillman, M. (2007), Reverse-engineering of polynomial dynamical systems, Adv. in Appl. Math., 39 (4), 477-489.
\bibitem{Jar2} Jarrah, A., Raposa, B., and Laubenbacher, R. (2007), Nested canalyzing, unate cascade, and polynomial functions, Physica D, 233, 167-174. 
\bibitem{Jar4} Jarrah, A. and Laubenbacher, R., and Veliz-Cuba, A. (2008), The dynamics of conjunctive and disjunctive Boolean networks, preprint available at: arxiv:0805.0275v1.
\bibitem{Kauf69} Kauffman,S.A. (1969), Metabolic stability and epigenesis in randomly constructed genetic nets, {\it J. Theor. Biol.} {\bf 22}, 437-467.
\bibitem{Kauf1} Kauffman, S., Peterson, C., Samuelsson, B., and Troein, C. (2004), Genetic networks with canalyzing Boolean rules are always stable, PNAS 101(49), 17102-17107.
\bibitem{Kino09} Kinoshita, S-i., Iguchi, K., and Yamada, H. S. (2009), Intrinsic properties of Boolean dynamics in complex networks, {\it J. Theor. Biol.} {\bf 256}, 351-369.
\bibitem{Laub} Laubenbacher, R. and Stigler, B. (2004), A computational algebra approach to the reverse engineering of gene regulatory networks, Journal of Theoretical Biology, Vol. 229, 523-537.
\bibitem{Lim} Lim, H. N., and van Oudenaarden, A. (2007), A multistep epigenetic switch enables the stable
inheritance of DNA methylation states, {\it Nature Genetics} {\bf 39 (2)}, 269-275.
\bibitem{Mol} Malloy, T. E., Butner, J. and Jensen, G. C. (2008), The emergence of dynamic form through phase relations in dynamic systems. Nonlinear Dynamics, Psychology, and Life Sciences, 12, 371-395.
\bibitem{Mey} Meyer, C.D. (2000), Matrix Analysis and Applied Linear Algebra, SIAM, Philadelphia, PA.
\bibitem{SR07} Saez-Rodriguez, J. {\it et al.} (2007), A logical model provides insights into T cell receptor signaling, {\it PLoS Computational Biology} {\bf 3(8)}: e163.
\bibitem{Shmu02} Shmulevich, I., Dougherty, E.R., Kim, S., and Zhang,  W. (2002), Probabilistic Boolean Networks: a rule-based uncertainty model for gene regulatory networks. Bioinformatics 18(2), 261-274. 
\bibitem{SLK99} Sol\'{e}, R.V., Luque, B., and Kauffman, S. (1999), Phase transition in random networks with multiple states, available online, arXiv:adap-org/9907011v1
\bibitem{Tamu1} Tamura, T. and Akutsu, T. (2008), Algorithms for Singleton Attractor Detection in Planar and Nonplanar AND/OR Boolean Networks, preprint, to appear in Mathematics in Computer Science.
\bibitem{Zhan} Zhang, S-Q., Hayashida, M., Akutsu, T., Ching, W-K., and Ng, M. K. (2007), Algorithms for finding small attractors in Boolean networks. EURASIP Journal on Bioinformatics and Systems Biology, doi:10.1155/2007/20180.
\bibitem{Zou10} Zou, Y. M. (2010), Modeling and analyzing complex biological networks incooperating experimental information on both network topology and stable states, {\it Bioinformatics} {\bf 26(16)}, 2037-2041.
\bibitem{Zou11} Zou, Y. M. (2011), Dynamics of Boolean networks, {\it Discrete and Continuous Dynamical Systems, Series S} {\bf 4(6)}, 1629-1640.
\bibitem{Zou12} Zou, Y. M. (2012), Characterizing bistable Boolean networks, ICBBE, 517-520.  
\end{thebibliography}
\end{document}